\documentclass[runningheads]{llncs}
\usepackage[T1]{fontenc}
\usepackage{graphicx}
\usepackage{todonotes}

\usepackage{algorithm}
\usepackage{algpseudocode}
\usepackage{mathrsfs}
\usepackage{mathtools}
\usepackage{pifont} 

\usepackage{tikz}
\usetikzlibrary{decorations.markings}
\usetikzlibrary{graphs,arrows.meta,decorations.pathreplacing}
\usetikzlibrary{positioning}
\usepackage{amsmath}
\usepackage{amssymb}
\usepackage{dirtytalk}
\usepackage{nicematrix}

\newcommand{\historyset}{\mathcal{H}^*(\mathbb{P}, \mathbb{T}, \mathbb{O}, \mathbb{V})}
\newcommand{\infhistoryset}{\mathcal{H}^\omega(\mathbb{P}, \mathbb{T}, \mathbb{O}, \mathbb{V})}

\newcommand{\rb}{\xrightarrow{rb}}
\newcommand{\ssrel}{\approx_{ss}}
\newcommand{\opa}{a}
\newcommand{\opb}{b}
\newcommand{\opc}{c}
\newcommand{\opd}{d}
\newcommand{\ope}{e}
\newcommand{\opf}{f}

\newcommand{\oph}{h}

\newcommand{\dsucp}[1]{\succrl{#1}{rb}}
\newcommand{\dsuc}{\dsucp{}}
\newcommand{\patharrow}{\!\mathrel{\rlap{\hskip 0.55em $\rightarrow$}\rightarrow}\;}
\newcommand{\algoneset}{\mathcal{G}^*}

\newcommand{\algone}{\textsl{Alg}}

\newcommand{\ar}{\xrightarrow{ar}}
\newcommand{\vis}{\xrightarrow{vis}}
\newcommand{\visc}{\xrightarrow{visc}}
\newcommand{\visrb}{\xrightarrow{visrb}}

\tikzset{
  withcaps/.style={
    postaction={
      decorate,
      decoration={
        markings,
        mark=at position 0 with {\draw[-] (0pt,-5pt) -- (0pt,5pt);},
        mark=at position 1 with {\draw[-] (0pt,-5pt) -- (0pt,5pt);}
      }
    }
  }
}

\newcommand{\eqdef}{\stackrel{\text{def}}{=}}
\newcommand{\varopa}{\mathsf{a}}
\newcommand{\varopb}{\mathsf{b}}
\newcommand{\varopc}{\mathsf{c}}
\newcommand{\arpred}{\stackrel{\mathsf{ar}}{\longrightarrow}}
\newcommand{\vispred}{\stackrel{\mathsf{vis}}{\longrightarrow}}
\newcommand{\succrl}[2]{\mathrel{\xrightharpoonup{#2}}_{#1}}
\newcommand{\logrl}[1]{\xrightarrow{\mathsf{#1}}}
\newcommand{\sessionorder}{\xrightarrow{so}}
\newcommand{\Enc}{\mathsf{Enc}}

\newcommand{\Types}{\mathbb{T}}
\newcommand{\Values}{\mathbb{V}}
\newcommand{\Objects}{\mathbb{O}}
\newcommand{\succs}{\mathsf{succs}}

\begin{document}

\title{HistMSO: a Logic for Reasoning about Consistency Models with MONA}

\author{Isabelle Coget\inst{1} \and
Etienne Lozes\inst{2}\orcidID{0000-0001-8505-585X}}

\authorrunning{I. Coget and E. Lozes}

\institute{Institut Polytechnique de Paris, Palaiseau, France \and
Université Côte d'Azur, CNRS, I3S, France}

\maketitle

\begin{abstract}
Reasoning about consistency models for replicated data systems is a challenging task that requires a deep understanding of both the consistency models themselves and a large part of human inputs in mechanized verification approaches.

In this work, we introduce an approach to reasoning about consistency models for replicated data systems. 
We introduce HistMSO, a monadic second-order logic (MSO) for histories and abstract executions, the formal models of executions of replicated data systems introduced by Burckhardt. 
We show that HistMSO can express 39 out of 42 consistency models from Viotti and Vukolic hierarchy. 
Moreover, we develop a method for reducing HistMSO satisfiability and model-checking to the same problems for MSO over words. 
While doing this, we leverage the MONA tool for automated reasoning on consistency models.

\end{abstract}

\section{Introduction\label{sec:intro}}
Replicated data systems consist of multiple agents (or processes) that maintain copies (replicas) of shared data objects. 
These agents perform operations, such as reads and writes, on the replicas. Due to network delays and the asynchronous nature of distributed systems, these operations may not be immediately visible to all agents, leading to potential inconsistencies among replicas. 
The replicated data system abstraction is widely used in various applications, including distributed databases, cloud storage systems, or fully decentralized storage systems such as blockchains or peer-to-peer networks. To some extent, consistency models are also related with weak memory models employed in hardware design~\cite{Burckhardt13}.
An important aspect of replicated data systems is that they often prioritize availability and partition tolerance over strong consistency guarantees. 
Since the strongest consistency model, linearisability~\cite{HerlihyW90}, is too costly to implement in many practical applications, replicated data systems often implement weaker consistency models (often called \emph{isolation level} in the database community). Monotonic reads, monotonic writes, read my writes, eventual consistency, etc are among the most popular ones of the lush jungle of consistency models. Viotti and Vukoli\'{c}~\cite{ViottiVukolic} inventoried 42 consistency models in 2016, and formally established the complex hierarchy they form (see Figure~1 in~\cite{ViottiVukolic}).

Following Burckhardt~\cite{PrinciplesOfEventualConsistency}, a consistency model is a property that an execution of a replicated data system must satisfy. The model of execution of a replicated data system introduced by Burckhardt, called a \emph{history}, bases on a collection of \emph{operations}, each located on a precise process. Each operation also lasts on a time interval that may partially overlap the time interval of another operation located on a different process. 
Like a Gantt diagram, a history therefore combines a partial order on operations and a continuous model of time. The history model is however often too simple for grounding a consistency model. 
Indeed, several consistency models rely on a logical justification of the way how the operations that compose a given history are related to each others. This justification may contradict the partial order and the timings of the operations. Burckhardt~\cite{PrinciplesOfEventualConsistency} therefore also introduced \emph{abstract executions}, that enrich histories with two binary relations among operations, called \emph{visibility} and \emph{arbitration}. 

In this work, we address the problem of automated reasoning on consistency models. Due to the combination of partial orderness and continuous time we just mentioned, it is not obvious which verification models, techniques and tools are the most appropriate for histories and abstract executions (see discussion in conclusion).
Our proposal in this work is to leverage MONA~\cite{MONA},  and more generally  MSO-to-automata translations, for automated reasoning on consistency models. MONA models, which are words over a finite alphabet, do not have a lot in common with histories. As a consequence, this requires a significant work for encoding the latter in the former. The choice of MONA also enforces some finiteness limitations, like a bounded, fixed, number of processes, or finitely many values (see conclusion for a detailed discussion on the limitations), which might be compensated by further applications of the automata-based approach (see also discussion in conclusion).

This work stems from an attempt to draw a parallel between communication models and consistency models. This work indeed is partly influenced by previous works of the second author on tree decompositions of message sequence charts for some weakly-synchronous communication models~\cite{DBLP:journals/pacmpl/GiustoFLL23}. Taking a step back in abstraction, we may address the question of automated reasoning on consistency model by transferring the tree decomposition techniques developed for message-passing concurrency to the realm of replicated data systems.

To sum up, we make the following contributions.
\begin{enumerate}
    \item We introduce HistMSO, the monadic second order logic of histories and abstract executions.
    \item We show that HistMSO can express 39 out of the 42 consistency models studied by Viotti and Vukoli\'{c}~\cite{ViottiVukolic}; remarkably, this result is quite straightforward, and follows from a translation of the meta logic used by Viotti and Vukoli\'{c} into HistMSO, based on rather standard MSO encoding techniques for expressing transitive closures, finiteness of sets defined by set comprehension, etc. 
    \item We show that the HistMSO theory of histories is decidable, by reducing the satisfiability problem of HistMSO to the one of MSO over infinite words. Moreover, our reduction also leverages MONA~\cite{MONAguide} as a practical theorem prover for the (weak) HistMSO theory of finite histories.
    \item We introduce $k$-transient visibility, a restriction on the visibility relation, and we show that the HistMSO theory of real-time, $k$-transient abstract executions is decidable.
    \item We show that the partial order defined by the timings of the operations of a given history is the transitive closure of a graph whose cutwidth is bounded by the square of the number of processes; since the treewidth is bounded by the cutwidth, this result induces in particular a tree decomposition of histories bounded by the number of processes and draws a parallel with the second author's previous work.
\end{enumerate}

\paragraph{Outline} Section~\ref{sec:background} recalls the models of histories and abstract executions. In Section~\ref{sec:histmso} we introduce HistMSO logic and show on a few representative examples how to axiomatise the majority of consistency models studied by Viotti and Vukolic. Section~\ref{sec:translation-to-mona} develops the translation of HistMSO to MONA, first for histories, then for $k$-transient real-time abstract executions. Section~\ref{sec:cutwidth} constructs a graph from an history whose transitive closure captures the returns before partial order over operations, and establishes a bound on the cutwidth of the graph that only depends on the number or processes.

\subsubsection{Related works}
Logical axiomatisations of communication models can be found in~\cite{DBLP:conf/fm/ChevrouH0Q19,DBLP:journals/pacmpl/GiustoFLL23,ENGELS2002253}.
Recent works on histories try to remove the arbitration relation from abstract executions~\cite{DBLP:journals/corr/abs-2510-21304}.
MONA was already applied to a wide range of applications, including reactive systems~\cite{KlaNieSun:casestudyautver}, pointer programs~\cite{JenJoeKlaSch:AutVerPoinProgMonSecOrdLog}, hardware verification~\cite{BasKla:BeyondFiniteHardware}, or parsing~\cite{YakYak99}.
There are several other automated verification techniques and tools that could have been considered instead of MONA, each with its own drawbacks (MONA being the finiteness assumptions we already mentioned).
Several model-checking tools are well suited for discrete time, with an interleaving semantics of parallelism (like SPIN~\cite{holzmann2007design} or TLA+~\cite{yu1999model}, to quote a very few), but they miss continuous-time.
Timed automata~\cite{alur1994theory} and the model-checking tools based on them~\cite{bengtsson1995uppaal,bozga1998kronos}, on the other hand, do not handle any partial order beyond the total order of events defined by their timings.
Automated theorem provers~\cite{kovacs2013first,weidenbach2009spass} and SMT solvers~\cite{de2008z3} are quite versatile, but with a limited support of quantifiers, specifically when quantifying over a binary relation. 
Interval Temporal Logics aim at automated reasoning on Gantt diagrams, but the strong theoretical results developed in this research community are not so well tools supported. 
That said, it is unclear to us whether some of these techniques, and the recent advances in each of these research communities, could not compete with our MONA-based approach, either on satisfiability and automated reasoning, or on runtime, distributed monitoring and model-checking.

\section{Background\label{sec:background}}

\begin{figure}[t]
\begin{center}
\resizebox{0.5\textwidth}{!}{
\begin{tikzpicture}[x=2cm, y=1cm,every node/.style={font=\large}]

	\node[anchor=east] at (-0.1, 0) {\small \textit{Proc. 1}};  % Adjust the x-coordinate as needed
  \node[anchor=east] at (-0.1, -1) {\small \textit{Proc. 2}}; % Label for second set  

  % First line of bars
  \draw[withcaps] (0,0) -- (2,0) node[midway, above] {$\opa$};
  \draw[dashed] (2,0) -- (2.5,0);
  \draw[withcaps] (2.5,0) -- (3,0) node[midway, above] {$\opb$};
  \draw[dashed] (3,0) -- (3.5,0);
  \draw[withcaps] (3.5,0) -- (4,0) node[midway, above] {$\opc$};

  % Second line of bars 
  \draw[dashed] (0,-1) -- (0.25,-1); 
  \draw[withcaps] (0.25,-1) -- (0.6,-1) node[midway, above] {$\opd$};
  \draw[dashed] (0.6,-1) -- (1,-1);
  \draw[withcaps](1,-1) -- (2.75,-1) node[midway, above] {$\ope$};
  \draw[dashed] (2.75,-1) -- (3.75,-1);
  \draw[withcaps](3.75,-1) -- (4,-1) node[midway, above] {$\opf$};

\end{tikzpicture}
}
\end{center}
\caption{\small \slshape A history on two processes, with operations $\opa$, $\opb$, $\opc$, $\opd$, $\ope$ and $\opf$}
\label{fig:ex-history}
\end{figure}
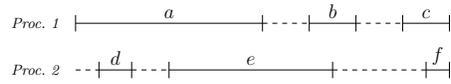

Tracking the sequence of computed actions in replicated data systems can be achieved through various methods. In this paper, we base our approach on the notion of histories defined in \cite{ViottiVukolic}. A history of a program execution is defined as a set of \textsl{operations} that represent all the actions performed during the computation. 

\subsubsection{Histories} can be visualized as timelines, where each process is represented by a horizontal line aligned to a global clock (see Figure~\ref{fig:ex-history}). The execution of an operation is shown as a solid segment, while dashed segments indicate periods when the process is idle. These timelines are read from left to right. 
 Each operation has the following attributes: the process running the operation, the invocation time, the return time, the type\footnote{following \cite{ViottiVukolic} terminology; however, it might be more accurate to think about it as the \emph{method name} of the replicated object; in this paper we consider that operations are of type either $read$ or $write$.} of operation, the object on which the operation takes effect, the input value of the operation and the output value of the operation. 
Since our objective is to encode these histories as words over a finite alphabet, for analysis using MONA, we make the assumption that the non-temporal attributes belong to finite sets of possible values. Without this finiteness assumption, the encodings presented in this paper would only work with an infinite alphabet.

Thereby we assume a set $\mathbb{P}$ of processes from which operations are launched, a set $\mathbb{T}$ of types of operations, and a set $\mathbb{O}$ of objects, which are the data to be modified or read by operations (objects can be registers or variables for instance). Also, we define the set of values as $\mathbb{V}\cup\{\nabla,\Phi\}$ where $\mathbb{V}$ is the set of values that can be read or written on objects. Symbol $\Phi$ represents a non-relevant or not needed value. For instance, we can claim that for a write operation, the input value is the value to be written, and the output value is $\Phi$; as a write operation does not necessarily need to return anything. Finally, $\nabla$ stands for the \say{output value} of an operation that never returns. Concretely, $\nabla$ represents the absence of a return value due to a never ending operation. 

As explained above we assume $\mathbb{P},\mathbb{T},\mathbb{O},\mathbb{V}$ finite.

Invocation and return times are defined in $\mathbb{R}_{>0}\cup \{+\infty\}$. We include $+\infty$ to represent operations that never return. Assigning $+\infty$ as a return time reflects that, beyond a certain point, the operation is considered to never terminate; invocation time however is always finite.
These timestamps refer to an ideal and global notion of time that we use to reason
about histories a posteriori, although it may not be accessible by processes during executions.

We recall the attributes from \cite{ViottiVukolic} along with their respective domains. The set of operation attributes is 
$$
\mathsf{Attributes} = \{proc, stime, rtime, type, obj, ival, oval\}
$$
An operation $o$ is thus defined as a record containing these attributes, with values taken from the following domains 
(throughout the paper, we use the notation $o.a$ to denote the value of a given attribute $a$ of an operation $o$):

\begin{itemize}
	\item $o.proc\in \mathbb{P}$ is the process on which the operation $o$ is executed
	\item $o.stime\in \mathbb{R}^+$ is the invocation time of $o$
	\item $o.rtime\in \mathbb{R}^+\cup \{+\infty\}$ is the return time of $o$, $+\infty$ if the operation never returns (we assume that for any operation $o$, it holds that $o.stime<o.rtime$)
	\item $o.type\in \mathbb{T}$ is the type of operation $o$
	\item $o.obj\in \mathbb{O}$ is the object on which the operation $o$ takes effect
	\item $o.ival\in \mathbb{V}$ is the input value of $o$, $\Phi$ if the operation is a $read$
	\item $o.oval\in \mathbb{V}$ is the output value of $o$, $\Phi$ if the operation is a write, and $\nabla$ if the operation does not return
\end{itemize}

Two operations $\opa$ and $\opb$ are said to be \emph{concurrent} if their execution times overlap, $i.e.$ if $\opa.stime\leq\opb.rtime$ and $\opb.stime\leq\opa.rtime$. Otherwise, the operations are \emph{sequential}. Note that if an operation never returns, it is concurrent with all operations that start after its invocation time. We assume in this paper that two concurrent operations cannot run on the same process.\footnote{This is important for the translation of histories to words over a \underline{finite} alphabet we present later. We could however relax this assumption by bounding the number of concurrent operations per process, which would serve modeling multi-threaded processes or operation/method calls inside an operation.}
We say that $\opa$ \emph{returns-before} $\opb$,  denoted as $\opa\rb\opb$, if $\opa$ returns before $\opb$ starts, i.e  $\opa.rtime<\opb.stime$. 
We say that $\opa$ is in the \emph{same session} as $\opb$, denoted as $\opa\ssrel\opb$, if $\opa$ and $\opb$ are running on the same process, i.e. $\opa.proc=\opb.proc$; we also say that $\opa$ and $\opb$ are in the \emph{session order}, $\opa\sessionorder\opb$, if $\opa\ssrel\opb$ and $\opa\rb\opb$.

\begin{example}
  In Figure~\ref{fig:ex-history}, operations $a$ and $d$ are concurrent, whereas operation $a$ returns before operation $f$; operations $a$ and $b$ are in the same session.
\end{example}

We say that two operations $\opa$ and $\opb$ are extremity-disjoint if $\opa.stime$, $\opa.rtime$, $\opb.stime$ and $\opb.rtime$ are pairwise distinct. 

\begin{definition}[History]\label{def:history}
  A history over $\mathbb{P}$, $\mathbb{T}$, $\mathbb{O}$, and $\mathbb{V}$ is a finite set $H$ of operations such that for every pair of distinct operations $\opa,\opb\in H$, (1) $\opa$ and $\opb$ are extremity-disjoint, and (2) if $\opa\ssrel \opb$, then $\opa$ and $\opb$ are sequential.
\end{definition}

A history is finite if it contains a finite number of operations.
The set of finite histories over $\mathbb{P}$, $\mathbb{T}$, $\mathbb{O}$, and $\mathbb{V}$ is denoted by $\historyset$.
An $\omega$-history is an infinite countable set of operations with no \emph{Zeno behavior}\footnote{We could also relax this assumption and translate arbitrary infinite histories to transfinite words indexed by ordinals, see e.g~\cite{demri-nowak-05-transfinite-words}.}: for all time $t\in \mathbb{R}^+$, there are finitely many operations that start before time $t$.
The set of $\omega$-histories over $\mathbb{P}$, $\mathbb{T}$, $\mathbb{O}$, and $\mathbb{V}$ is denoted by $\infhistoryset$.

In Section~\ref{sec:graphs-finite-histories}, we will study the oriented graph $(H,\rb)$ of the operations of a history $H$. The $\rb$ relation will be "too redundant" for the purpose of bounding some complexity measure of the graph, there we will consider a more restrictive relation, which we call \textsl{direct succession}.

\begin{definition}[Direct Successor]\label{def:successor}
Let $H\in\historyset$ and $\mathrel{\xrightarrow{rel}}$ a binary relation over $H$. We say that an operation $\opb$ is the direct successor of $\opa$ on process $p$, written $\opa \succrl{p}{rel} \opb$, if $\opb$ is the first operation on process $p$ such that $\opa\mathrel{\xrightarrow{rel}} \opb$. 
\end{definition}

Note that the relation $\opa\dsucp{p}\opb$ holds information about the process of $\opb$ (namely, $p$), but does not reveal any information about the process on which $\opa$ is executed. Additionally, we say that $\opb$ is a direct successor of $\opa$, written $\opa\dsuc\opb$, if there exists a process $p$ such that $\opa \dsucp{p} \opb$.

\subsubsection{Abstract executions} are an enrichment of histories with two relations over operations: arbitration and visibility. Observe that a history does not contain enough information about operations to assess its consistency. Indeed, if two operations are concurrently executed, there is no way of knowing which operation should be considered as "logically before" the other. Thus, histories are traditionaly enriched with two relations over operations: arbitration and visibility.\footnote{
Quoting \cite{PrinciplesOfEventualConsistency}, page~46:
\emph{Visibility tells us about the relative timing of update propagation and operations. It is an acyclic relation. If an operation $\opa$ is visible to $\opb$ (written $\opa\vis\opb$), it means that the effect of $\opa$ is visible to the client performing $\opb$. In a system where updates are communicated by messages, this may mean that the message about operation $\opa$ reached the client that performed operation $\opb$ before the operation $\opb$ was performed.
 Arbitration is used to indicate how the system resolves update conflicts, $i.e.$ how it handles concurrent updates that do not commute. It is a total order on operations.  If an operation $\opa$ is arbitrated before $\opb$ (written $\opa\ar\opb$), it means that the system considers the operation $\opa$ to happen earlier than operation $\opb$.}
}

\begin{figure}[t]
\begin{center}
\begin{tikzpicture}[x=2cm, y=1cm,every node/.style={font=\large}]

	\node[anchor=east] at (-0.1, 0) {\small \textit{Proc. 1}};  % Adjust the x-coordinate as needed
  \node[anchor=east] at (-0.1, -1) {\small \textit{Proc. 2}}; % Label for second set  

  % First line of bars
  \draw[withcaps] (0,0) -- (2,0) node[midway, above] {\tiny x.write(42)};
  \draw[dashed] (2,0) -- (2.5,0);
  \draw[withcaps] (2.5,0) -- (3,0) node[midway, above] {\tiny x.read()$\to$17};
  \draw[dashed] (3,0) -- (3.5,0);
  \draw[withcaps] (3.5,0) -- (4,0) node[above, pos=0.7] {\tiny x.read()$\to$42};
  % Second line of bars 
  \draw[dashed] (0,-1) -- (0.25,-1); 
  \draw[withcaps] (0.25,-1) -- (0.6,-1) node[midway, below] {\tiny y.write(true)};
  \draw[dashed] (0.6,-1) -- (1,-1);
  \draw[withcaps](1,-1) -- (2.75,-1) node[midway, below] {\tiny  x.write(17)};
  \draw[dashed] (2.75,-1) -- (3.75,-1);
  \draw[withcaps](3.75,-1) -- (4,-1) node[midway, below] {\tiny x.read()$\to$42};

  % Arrows for program order
  \draw[->, thick] (1.5, -1) -- node[above, sloped] {\tiny $ar$} (1.6, 0); 
  \draw[->, thick] (2.5, -1) -- node[above, sloped, pos=0.4] {\tiny $vis$} (2.6, 0); 
  \draw[->, thick] (1.8, 0) -- node[above, sloped] {\tiny $vis$} (3.85, -1); 
  \draw[thick] (1.7, 0) edge [->, bend left=30] node[above, sloped] {\tiny $vis$} (3.6, 0); 

\end{tikzpicture}
\end{center}
\caption{\small \slshape An abstract execution over the history of Figure~\ref{fig:ex-history}, with arbitration and visibility relations}
\label{fig:ex-abstract-execution}
\end{figure}
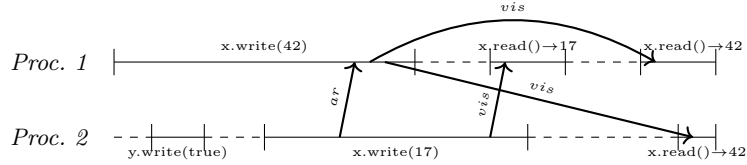

\begin{definition}[Abstract execution~\cite{ViottiVukolic}]\label{def:abstract-execution}
  An abstract execution is given by a triple \mbox{$(H,\ar,\vis)$} where $H$ is a history, $\ar$ is an arbitration relation, and $\vis$ is a visibility relation such that:
  \begin{itemize}
    \item $\ar$ is a total order over $H$;
    \item $\vis$ is an acyclic relation\footnote{with no look-ahead, see Definition~\ref{def:no-look-ahead} below} over $H$.    
  \end{itemize}
\end{definition}

\begin{example}
  In Figure~\ref{fig:ex-abstract-execution}, we illustrate an abstract execution over the history of Figure~\ref{fig:ex-history}. The arbitration relation $\ar$ is represented by arrows going from an operation to another one. We only depict an arrow $a\ar b$ if $a$ and $b$ are on different processes and if $b$ is the first operation after $a$ in the arbitration order on process $p = b.proc$ (we later write $\opa \succrl{p}{ar} \opb$ when this is the case). The visibility relation $\vis$ is represented similarly, as well as the direct succession relation $\opa\succrl{p}{vis}\opb$.
\end{example}

Note that Burckhardt's model of abstract executions is quite general, and does not make any assumption on the relations between $\ar$, $\vis$, and $\rb$ beyond the ones explicitly stated in Definitions~\ref{def:history} and~\ref{def:abstract-execution}. To conclude this section, we try to hint why abstract executions are defined in this way, and introduce a new property, absence of look-ahead, that seems to us rather reasonable to assume for all abstract executions.

\begin{remark}\label{rem:arbitration-visibility}
  Visibility does not imply arbitration. This property is sometimes enforced by some consistency models (for instance, by the $\mathsf{SingleOrder}$ property used to define linearizability, see below), but it is not the case for all consistency models. Arbitration, that corresponds to a global consensus on the order of operations, may deviate from visibility in some scenarios: for instance, with three processes and arbitration based on majority, if two processes perform \texttt{x.write(0)} and the third one performs \texttt{x.write(1)}, they may collectively arbitrate \texttt{x.write(1)} before the two \texttt{x.write(0)} in order to reflect majority voting and "drop" the \texttt{x.write(1)}, even if both \texttt{x.write(0)} were visible to the third process when he performed \texttt{x.write(1)}.
\end{remark}

\begin{remark}
  Visibility does not imply returns-before. Indeed, an operation can be visible to another operation even if it has not yet returned. In Figure~\ref{fig:ex-abstract-execution}, for instance, \texttt{x.write(17)} on process 2 is visible to \texttt{x.read(17)} on process 1, even if \texttt{x.write(17)} has not yet returned when \texttt{x.read(17)} started. Intuitively, process 2 sent a message about \texttt{x.write(17)} to process 1 before \texttt{x.write(17)} finished, and this message reached process 1 before \texttt{x.read(17)} finished.
\end{remark}

\begin{remark}
  At a process-local scale, visibility does not imply session order. This is weird if we stick to the intuition that visibility is related to message passing, because it would mean a process can receive a message from itself that it will send in the future. If we do not stick to the message passing intuition, we can think about visibility in a broader sense, and allow processes to rely on some look-ahead on the operations they will execute. 
\end{remark}

For simplicity, we later exclude this look-ahead behavior\footnote{we could probably cope with a finite look-ahead, but this would complicate the translation to MSO, and it is unclear how relevant it would be} and assume the following property (that expresses the absence of look-ahead at the process-local scale and further restricts visibility accross distinct processes, seen as a message passing relation, to allow receiving messages that are sent in the future).

\begin{definition}[No look-ahead]\label{def:no-look-ahead}
  We say that an abstract execution $(H,\ar,\vis)$ satisfies the no look-ahead property if for all operations $\opa$ and $\opb$, if $\opa\rb\opb$, then $\opb\not\vis\opa$.
\end{definition}

In the rest of the paper, we implicitly assume that this property is part of Definition~\ref{def:abstract-execution}.
We conclude with two further remarks about the fact that $\ar$, $\vis$, and $\rb$ relations are independent of each other.

\begin{remark}
  Session order (and more generally, $\rb$) does not imply visibility. For instance, in Figure~\ref{fig:ex-abstract-execution}, on process 1, \texttt{x.write(42)} is not visible to \texttt{x.read(17)}.
  This would however be enforced by consistency models such as $\mathsf{ReadYourWrites}$, that require that an operation is visible to all subsequent operations in the same session.
\end{remark}

\begin{remark}
  Session order (and more generally, $\rb$) does not imply arbitration. This is the case for instance when a write operation is dropped by the system, arbitrating it before other writes, see discussion in Remark~\ref{rem:arbitration-visibility} above.
\end{remark}

\section{A Logical Formalism for Consistency Models\label{sec:histmso}}
In this section we introduce HistMSO, a logic that allows to express properties of 
histories and abstract executions in distributed systems. 
We first define the syntax of the logic, then we show how to express consistency models in this logic.
\subsection{Definition of HistMSO\label{sec:definition-logic}}

We define HistMSO, the Monadic Second-Order (MSO) logic of histories and abstract executions as expected. MSO logic is the restriction of second-order logic where second-order quantifications are limited to quantifications over sets (of operations). Formally, the grammar of the MSO formulas $\phi$ over histories and abstract executions is the following, with $\varopa$ and $\varopb$ ranging over first order variables (interpreted as operations), and $A$ ranging over second order variables (interpreted as sets of operations),
and $attr$ ranges over operation attributes.

$$
\begin{array}{ccl}
    t & ::= & \varopa.stime \mid \varopa.rtime \\
    f & ::= & \varopa.attr=\varopb.attr  \mid  t < t \mid \varopa.proc = p\\
    g & ::= & \varopa.ival=\Phi \mid \varopa.oval=\Phi \mid \varopa.oval=\nabla\\
    \phi & ::= & f \mid g \mid \phi\vee\phi \mid \neg \phi \mid \forall \varopa.\ \phi \mid \forall A.\ \phi \mid \varopa \in A \\
            & \mid & \varopa \arpred \varopb \mid \varopa \vispred \varopb 
\end{array}
$$

The formulas $\phi_1\wedge \phi_2$ (conjunction), $\phi_1\Rightarrow \phi_2$ (implication), $\exists \varopa.\ \phi$ (first order quantification), 
guarded quantification $\forall \varopa \in X.\ \varphi$,
and so on are defined as macros as expected.
Formulas that do not contain $\vispred$ and $\arpred$ are interpreted over histories, while formulas that do contain these relation symbols are interpreted over abstract executions. The satisfaction relation 
$(H,\ar,\vis),\upsilon \models \phi$ between an abstract execution $(H,\ar,\vis)$, a variable assignment $\upsilon$, and a formula $\phi$ is defined as usual, with the obvious interpretations for the $\arpred$ and $\vispred$ predicates.
The set of free variables of a formula $\phi$ is defined as usual, and a formula is closed if the set of free variables is empty.

\begin{remark}
    The relations $\rb$ and $\ssrel$ introduced in Section~\ref{sec:graphs-finite-histories} can be defined in the logic. For instance, $\varopa\logrl{rb}\varopb \eqdef \varopa.rtime<\varopb.stime$.
    The direct successor relations $\dsuc$, $\dsucp{p}$, $\succrl{p}{ar}$, etc introduced in the previous section can also be defined in the logic.
    For instance, 
    $$
        \varopa\succrl{p}{\mathsf{rb}}\varopb \quad \eqdef \quad \varopb.proc=p \;\land\; \varopa \logrl{rb}\varopb \;\land\;
            \neg\exists\varopc.\varopc\ssrel\varopb\land  \varopa\logrl{rb}\varopc  \land  \varopc\logrl{rb}\varopb
    $$
\end{remark}

\subsection{Expressing Consistency Models in HistMSO}
In this section we exercise the expressive power of HistMSO by formally defining several consistency models from Viotti and Vukolic classification~\cite{ViottiVukolic}. Our finding is that almost all of the 42 models formalized by Viotti and Vukolic (see~\cite{ViottiVukolic} Figure~1) can be expressed in our logic. The only exceptions are the few models based on exact timing constraints such as \textsl{Timed Visibility}\cite{DBLP:conf/spaa/SinglaRH97}\footnote{even those models could be approached by moving to a discrete time semantics}. Due to space constraints, we only present a very short selection of consistency models here, showing by some examples how to reformulate in HistMSO the meta logic used by Viotti and Vukolic. The reader may consider looking first at Viotti and Vukolic
formalisation of the consistency models in their meta logic, see page 41 of~\cite{ViottiVukolic}.

\subsubsection{Return value consistency} ensures that all operations return the appropriate and correct values. In this work, this means that any read operation should output the most recent value written to the corresponding object by a visible operation. The following formula intuitively captures this behavior: it holds if all \textsl{writes} return the empty value $\Phi$, and if every \textsl{read} either accesses an object that has not been written to before, or retrieves the last (in the sense of arbitration) value visibly written to that object.

Before expressing these with formulas, let us first recall the notion of context of an operation, introduced by Viotti and Vukolic. An operation $\opb$ is part of the context explaining a read operation $\opa$ if it is a write on the same object that is visible to $\opa$.

$$
\opb\in\mathsf{ctxt}(\opa) \quad \eqdef\quad \opb\vis \opa\wedge \opb.type=write\wedge \opb.obj=\opa.obj
$$

It becomes then easy to express return value consistency in HistMSO.

$$
    \begin{array}{rcl}
        \textsc{RVal} & \eqdef & \forall \opa.\; W(a) \wedge R(a)
        \\
        W(a) & \eqdef & \opa.type = write \Rightarrow \opa.oval = \Phi 
        \\ 
        R(a) & \eqdef & \opa.type = read \Rightarrow \opa.oval=\mathsf{lastWrite}(\mathsf{ctxt}(\opa)) 
        \\
        v=\mathsf{lastWrite}(C) & \eqdef & \forall \opb.\; \big(\opb\in C\wedge \opb.ival\neq v\big)\; \Rightarrow\; \exists \opc \in C.\; \opb\ar\opc
    \end{array}
$$

\subsubsection{Real time} guarantees that any two operations not concurrently executed are ordered by \textsl{arbitration} according to absolute time. In Viotti and Vukolic meta logic, this is expressed as $rb\subseteq ar$. In HistMSO, this becomes

$\textsc{RealTime} \eqdef \forall a.\forall b. a\xrightarrow{rb}b \Rightarrow a\xrightarrow{ar}b $.

\subsubsection{Linearizability} states that \say{each operation shall appear to be applied instantaneously at a certain point in time between its invocation and its response}~\cite{HerlihyW90}. Following Burckhardt~\cite{PrinciplesOfEventualConsistency}, Viotti and Vukolic formalize linearisability as the conjunct of return value consistency, real time, and a third notion, called single order:
$
 \textsc{Linearizability} \eqdef \textsc{SingleOrder}\land\textsc{RealTime}\land\textsc{RVal}.
$
Intuitively, single order tights toghether visibility and arbitration. In Viotti and Vukolic meta logic, this is expressed as
$
\exists H ' \subseteq \{ op \in H : op.oval= \nabla \} : vis = ar\setminus (H'\times H).
$
In HistMSO, this becomes
$$
\begin{array}{rcl}
\textsc{SingleOrder} & \eqdef & \exists X.( \forall x.x\in X\Rightarrow x.oval=\nabla )\land\\ 
&& \forall a.\forall b.(a\xrightarrow{vis}b \Leftrightarrow a\xrightarrow{ar}b\land\neg(a\in X))
\end{array}
$$  

\subsubsection{Quiescent Consistency} requires that if all objects stop receiving
updates (i.e., become quiescent), then the execution is equivalent to some sequential execution containing only complete operations. Although this definition resembles eventual consistency, it does not guarantee termination: a system that does not stop receiving updates will not reach quiescence, thus replicas convergence. Following Viotti and Vukolic\footnote{We slightly simplified their formula for presentation purposes. Their formula was $|H_{|wr}|<\infty \Rightarrow \exists C\in\mathcal C :\forall [f]\in H/{\approx_{ss}}:|\{op\in[f]:op.oval\neq\mathsf{lastWrite}(C)\}|<\infty
$. In this formula, $C$ is a context and $[f]$ is the set of operations of some process. Due to the finiteness of the set of processes, it is equivalent to $|H_{|wr}|<\infty \Rightarrow \exists C\in\mathcal C:\, |\{op\in H \mid op.oval\neq \mathsf{lastWrite}(C)\}|<\infty
$.}, we can express quiescent consistency in HistMSO based on set finiteness.

$$
\begin{array}{rcl}
    \textsc{QuiescentConsistency} & \eqdef & \mathsf{Finite}(\{\opa\mid\opa.type=write\})\; \Rightarrow  \; \mathsf{FiniteInconsistency}
    \\
    \mathsf{FiniteInconsistency} & \eqdef &
    \exists a.\ a.type=read \\ 
    && \wedge \mathsf{Finite}\big(\{b\mid b.oval\neq\mathsf{lastWrite}(\mathsf{ctxt}(a))\}\big)
\end{array}
$$

On arbitrary structures, set finiteness is not MSO definable. However, since $\omega$-histories are non-Zeno and involve a finite number of processes, a subset $O$ of an $\omega$-history $H$ is finite if and only if there is a time $t$ such that all operations of $O$ starts before $t$.
$$
\mathsf{Finite}(O)\eqdef (\neg\exists a\in O)\; \vee\; 
\exists \opa\in O.\; \forall \opb\in O.\;\opb.stime\leq \opa.stime
$$

All together, we have that quiescent consistency is expressible in HistMSO.

\section{Translation to MONA\label{sec:translation-to-mona}}
In this section we show how to translate HistMSO to MSO over natural numbers.
This translation allows us to leverage MONA's decision procedures to verify properties of finite histories and abstract executions in distributed systems. We first recall some background on MONA, then we present our translation for histories, and finally we discuss how to extend this translation to abstract executions.

\subsection{Background on MONA}

MONA deals with the following logical formalism:
$$
\varphi ::= x < y \mid x \in X \mid \neg \varphi \mid \varphi \land \varphi \mid \exists x.\varphi \mid \exists X.\varphi
$$
where $x,y$ are first-order variables ranging over natural numbers, and $X$ is a second-order variable ranging over sets of natural numbers. A model of a formula $\varphi$
with free (first and second order) variables $\upsilon_1,\ldots,\upsilon_n$ (the order of the enumeration matters) is a word over the alphabet $\{0,1\}^n$. Each position $i$ of the word corresponds to a natural number, and the $j$-th bit of the letter at position $i$ is 1 if and only if the $j$-th order-1 (resp. order 2) variable is interpreted as equaling to (resp. containing) $i$ (see~\cite{MONAguide} for more details).

Our goal in the following is to define an encoding function $\Enc$ that takes as input a finite history $H\in\historyset$ and outputs a word of bitvectors $\Enc(H)\in(\{0,1\}^n)^*$ (for some $n$ depending on the meta parameters $\mathbb{P}, \mathbb{T}, \mathbb{O}, \mathbb{V}$) and a translation $T(\varphi)$ of a formula $\varphi$ of the MSO logic of histories and abstract executions defined in Section~\ref{sec:definition-logic} into a MONA formula $T(\varphi)$ such that 
\begin{equation}\label{eq:correctness-of-translation}
H \models \varphi \qquad \iff \qquad \Enc(H) \models_{\tiny \mathsf{MONA}} T(\varphi)
\end{equation}

\subsection{Encoding histories to words}

Each letter of $Enc(H)$ represents the operations and relations of $H$ at a certain point in time: the word $Enc(H)$ forms a sequence of snapshots of $H$ over time, quite similar to a discrete sampling of a continuous signal. To define the encoding function we follow these steps: coding the timeline, then adding information about types and values.

\subsubsection{Coding the timeline}
Let us first fix some enumeration $\mathbb{P}=\{p_1,...,p_m\}$ of the (finite) set of processes.
For a finite history (resp. an $\omega$-history) $H$, and for a timestamp $t\in \mathbb{R}_{\geq 0}$,
let $\mathsf{snapshot}(H, t)\in\{0,1\}^{|\mathbb{P}|}$ denote the bitvector 
whose $i$-th coordinate is 1 if and only if process $p_i$ is active in $H$ at timestamp $t$;
for an increasing sequence $\mathbf{T}$ of timestamps $t_0<t_1<...$ we define
$\mathsf{snapshots(H,\mathbf{T})}$ as the sequence of snapshots $\mathsf{snapshot}(H, t_i)$.
The timeline encoding $\Enc_{tl}(H)$ is $\Enc(H,\mathbf{T}_H)$ where $\mathbf{T}$ denote the enumeration of all starting and ending times of operations in $H$.

\begin{example}
The construction is illustrated on Figure~\ref{fig:encodingExample1AllVectors}.
The history $H$ on this figure is run on three processes. 
We therefore have $Enc(H)\in (\{0,1\}^3)^*$. Each time an operation starts or returns, we take a snapshot. 
\end{example}

\begin{figure}[t]
\begin{center}
\resizebox{0.7\textwidth}{!}{
\begin{tikzpicture}[x=1cm, y=1.25cm,line width=2pt, every node/.style={font=\LARGE}]
  
  \node[anchor=east] at (-0.3, 0) {$p_1$};  
  \node[anchor=east] at (-0.3, -1) {$p_2$}; 
  \node[anchor=east] at (-0.3, -2) {$p_3$};  

  \draw[dashed,line width=1pt] (0,0) -- (19,0);
  \draw[dashed,line width=1pt] (0,-1) -- (19,-1);
  \draw[dashed,line width=1pt] (0,-2) -- (19,-2);

  % Horizontal operations
  \draw[withcaps] (2,0) -- (5,0) node[midway, above] {$\opa$};
  \draw[withcaps] (10,0) -- (14,0) node[pos=0.45, above] {$\opb$};
  \draw[withcaps] (15.5,0) -- (18,0) node[midway, above] {$\opc$}; 
  
  \draw[withcaps] (3,-1) -- (7,-1) node[midway, above] {$\opd$};
  \draw[withcaps] (12,-1) -- (17,-1) node[midway, above] {$\ope$};

  \draw[withcaps] (8,-2) -- (11,-2) node[midway, above] {$\opf$};

  % Vertical time lines at starts and ends
  \foreach \x/\t/\vone/\vtwo/\vthree in {
  0/T_0/0/0/0,
  2/T_1/1/0/0,
  3/T_2/1/1/0,
  5/T_3/0/1/0,
  7/T_4/0/0/0,
  8/T_5/0/0/1,
  10/T_6/1/0/1,
  11/T_7/1/0/0,
  12/T_8/1/1/0,
  14/T_9/0/1/0,
  15.5/T_{10}/1/1/0,
  17/T_{11}/1/0/0,
  18/T_{12}/0/0/0
  } {
    \draw[dashed,blue] (\x,-2) -- (\x,0) 
	node[below=10em,black] {$\begin{bmatrix} \vone \\ \vtwo \\ \vthree \end{bmatrix}$}
    node[below] at (\x,-2.20) {$\t$};
  }

\end{tikzpicture}
}
\end{center}
  \caption{A finite history $H$ and its timeline encoding $\Enc_{tl}(H)$}
\label{fig:encodingExample1AllVectors}
\end{figure}
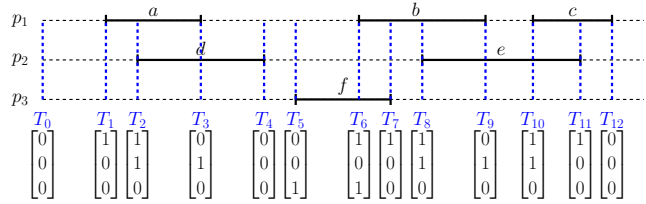

\begin{remark}
We assumed in Definition~\ref{def:history} that the histories we consider are such that every timestamp
occurs in at most one operation extremity, and no operation starts at $t=0$. 
This ensures that every two consecutive vectors in $\Enc_{tl}(H)$ differ at exactly one coordinate and help identify the process on which an operation starts or stop at a given timestamp (see for instance the formula $\mathsf{isStart}$ we later define).
\end{remark}

\subsubsection{Adding information about types, values, and objects}

For now the only information encoded are the start and return times. We want to enrich the encoding with the $type$, $ival$/$oval$ and $obj$ attributes. 
We will do so by adding coordinates to the vectors we just used to define $\Enc_{tl}(H)$. 
Remember that we assumed the meta parameters $\Types$, $\Values$ and $\Objects$ are finite sets. 
We can therefore assume that a type, a value, or an object identifier can be coded on a number of bits that can be determined in advance. 
For instance, we assumed above for simplicity that $\Types$ contains only \texttt{read} and \texttt{write}.
We can simply code a \texttt{read} as a 0 and a \texttt{write} as a 1. For $Values$, we can stick to their machine representation. To encode $\Objects$ we can fix an enumeration $o_1,\ldots o_n$ and represent the object identifier $o_i$ by the binary encoding of $i$ of $\left\lceil\log_2(n)\right\rceil$ bits.

In $\Enc_{tl}(H)$, two successive vectors differ in only one coordinate. This means, in particular, that when introducing a new operation $\opa$ on a vector $w$, we only need to encode data about $\opa$. 

Each vector of $Enc(H)$ is of the following form. The $m$ first coordinates are used as explained above in the definition of $\Enc_{tl}(H)$. The following coordinates encode the attributes of the operation that either starts or stop at this timestamp: the $m+1$ coordinate encodes the type of the operation (0 for $read$, 1 for $write$), the next $\left\lceil \log_2 |Values| \right\rceil$ coordinates are dedicated to encoding either the input or the output value\footnote{A $read$ operation does not have an input value and a $write$ operation does not have an output value}, and the last $\left\lceil \log_2 |Objects| \right\rceil$ to the encoding of the object identifier. The very first vector is the null vector.

\subsubsection{Translating formulas}

Finally, we sketch the definition of the translation $\varphi\mapsto T(\varphi)$
of MSO formulas over histories to MONA formulas. Each line $\ell$ 
of a bitvector of $\Enc(H)$ is associated with a second order variable $X_{\ell}$.
The translation otherwise maps each first order variable $\varopa$ (ranging over an operation) to a first order variable ranging over a column of $\Enc(H)$, that we also write $\varopa$. Similarly, second order variables are mapped to second order variables.
More precisely, we define the translation inductively as follows:
$$
\begin{array}{l}
T(\varphi\vee\psi) \eqdef T(\varphi) \vee T(\psi) 
\qquad
T(\neg\varphi) \eqdef \neg T(\varphi) 
\\
T(\exists \varopa.\ \varphi) \eqdef \exists \varopa.\ \mathsf{isStart}(\varopa)\wedge T(\varphi)
\qquad
T(\varopa\in A) \eqdef \varopa \in A
\\
T(\exists A.\ \varphi) \eqdef \exists A.\ (\forall \varopa.\ \varopa\in A\Rightarrow \mathsf{isStart}(\varopa))\wedge T(\varphi)
\end{array}
$$
where $\mathsf{isStart}(\varopa)$ is a MONA formula stating that the vector at position $\varopa$ corresponds to the start of an operation (that is, $a\neq 0$ and the only bit that differs between $a-1$ and $a$ is equal to $1$ at $a$).
The atomic formulas are translated as expected from the encoding defined above. For instance $T(\varopa.obj=p)= \varopa \in X_p$, where $X_p$ is the free variable associated with the line of $p$ in $\Enc_{tl}(H)$. The translation of $\varopa.tattr<\varopb.tattr'$ is $\exists \varopa',\varopb'.\ tattr(\varopa,\varopa')\wedge tattr'(\varopb,\varopb')\wedge \varopa'<\varopb'$, with
$stime(\varopa,\varopa')\eqdef \varopa=\varopa'$,
and $rtime(\varopa,\varopa')$ expressing that $\varopa'$ is the first position 
after $\varopa$ where the bit on the line of the process of $\varopa$ is null. We omit the full definition of the translation for brevity. We also omit the definition of the formula $\mathsf{isEncoding}$ that checks that a word is indeed an encoding of a history, i.e $w\models_{\mathsf{MONA}} \mathsf{isEncoding}$ if and only if there exists $H\in\historyset$ such that $w=\Enc(H)$.

\begin{theorem}
  For every closed formula $\varphi$ of HistMSO$\setminus\{\ar,\vis\}$, for every finite history $H\in\historyset$ (respectively $\omega$-history $H\in\infhistoryset$), it holds that
  $$
    H \models \varphi \qquad \iff \qquad \Enc(H) \models_{\mathsf{MONA}} T(\varphi)
  $$
\end{theorem}
\begin{corollary}
  The HistMSO theory of finite histories and the HistMSO theory of $\omega$-infinite histories are decidable.
\end{corollary}

\subsection{Extension to abstract executions}
In this section we explain how to extend the encoding defined for histories to abstract executions. It is not clear how to encode the relations $\vis$ and $\ar$ directly in the word encoding of an abstract execution, as they are not directly related to the timeline of operations. Therefore, we propose a few assumptions that can be made on these relations, and then exploit them to extend our encoding technique.

\subsubsection{Representing arbitration in real-time histories}

Remember that arbitration is \emph{any} total order on operations. Without any further assumption, the two relations $\rb$ and $\ar$ could be used to define a grid, and the logic would be undecidable even at first order.

However, in practice, arbitration is not completely independent of the timeline. 
We therefore make the assumption that the abstract executions we consider are 
real-time\footnote{Weaker forms of this assumption, like "boundedly deviating from real-time", could be worth exploring.} (see formula \textsc{RealTime} in the previous section).

With this restriction, the only pairs of operations $(\opa,\opb)$ for which our encoding should specify whether $\opa$ is arbitrated before $\opb$ or the converse are those pairs of concurrent operations. 
To encode this information, we add $|\mathbb{P}|-1$ extra bits to the vectors we used
in the definition of $\Enc(H)$. If $t$ is the starting time of an operation $\opa$ on process $p_i$, then the $j$-th extra bit (for $j\neq i$) indicates whether $\opa \ar \opb$, where $\opb$ is the last operation started on process the $j$th process (skipping $p_i$ if $j>i$) before time $t$.

This coding allows us to define a binary relation $\xrightarrow{arc}$ such that 
$\xrightarrow{arc}\cap\xleftarrow{arc}=\emptyset$ and $a(\xrightarrow{arc}\cup\xleftarrow{arc})b$ iff $a$ and $b$ are concurrent. We then just need to define
$\ar$ as $(\rb \cup \xrightarrow{arc})^+$, and assert in the formula $\mathsf{isEncoding}$ that $\ar$ is acyclic.

\subsubsection{Representing visibility under real-time and $k$-transient visibility assumptions}

Visibility is even more versatile and challenging to encode. We again rely on the real-time assumption to restrict the possible visibility relations, which enforces that $\forall \opa,\opb, \opa \rb \opb \Rightarrow \neg(\opb\vis\opa)$.
We can therefore decompose the visibility relation into two parts: the visibility between concurrent operations, which can be encoded as we did for arbitration, and the visibility between operations $\opa$ and $\opb$ such that $\opa\rb\opb$, two relations we denote respectively by $\visc$ and $\visrb$, with $\vis=\visc\cup\visrb$. We therefore add an extra bit to the vectors of $\Enc(H)$ to encode $\visc$ as we did for arbitration.

In order to encode $\visrb$, we need to make further assumptions on visibility. We cannot make too strong assumptions, as we want to be able to tell apart all consistency models defined in~\cite{ViottiVukolic}. For instance,
it would be problematic to make the assumption that once an operation $\opa$ is visible to another operation $\opb$, then it is also visible to all operations started after $\opb$ on the same process. Indeed, this assumption, known as PRAM, or FIFO, subsumes some prominent consistency models like \emph{monotonic reads}, \emph{monotonic writes}, or \emph{read my writes}.

Remember that $a\succrl{p}{rb} b_0$ if $b_0$ is the first operation on process $p$ that starts after the return time of $a$. Given a fixed $a$ and $p$, 
consider $b_0,b_1,\ldots,b_{i},\ldots$ such that $b_0\succrl{p}{rb} b_{1} \succrl{p}{rb} \cdots \succrl{p}{rb} b_{k-1}$. We say that $a$ is $k$-transient to process $p$ if for all $i\geq k$, $a\vis b_i$ iff $a\vis b_{k-1}$.
Intuitively, this means that the visibility of an operation $\opa$ to operations on a given process $p$ stabilises after a bound $k$ to either always visible or always invisible. 

\begin{definition}[$k$-transient visibility]
Let $k\geq 0$. An abstract execution $(H,\ar,\vis)$ has $k$-transient visibility if for all operations $\opa\in H$, for all processes $p\in \mathbb{P}$, $a$ is $k$-transient to process $p$.
\end{definition}

We therefore add an extra meta parameter $k$, and add $k\cdot |\mathbb{P}|$ bits to the vectors of $\Enc(H)$. These bits are used to encode, for each process timestamp
corresponding to the start of an operation $\opa$, the visibility of $\opa$ to the first $k$ operations started on each process $p$ after the return of $\opa$.
With this encoding, we can define a formula for $\visrb$ in MONA, and impose in the formula $\mathsf{isEncoding}$ that $\vis=\visc\cup\visrb$ is acyclic.

\begin{theorem}
  For every closed formula $\varphi$ of HistMSO for every finite (resp. $\omega$-infinite) abstract execution $(H,\ar,vis)$, it holds that
  $$
    H \models \varphi \qquad \iff \qquad \Enc(H) \models_{\mathsf{MONA}} T(\varphi)
  $$
\end{theorem}
\begin{corollary}
  The HistMSO theory of finite (resp. $\omega$-infinite) abstract executions with real-time arbitration and $k$-transient visibility is decidable.
\end{corollary}

\section{Cutwidth complexity of histories\label{sec:cutwidth}}
In this section we show that the partial order $\rb$ among operations of a
given \underline{finite} history is the transitive closure of a graph whose cutwidth is bounded
by the square of the number of processes; since the treewidth is bounded
by the cutwidth, this result induces in particular a tree decomposition of
histories bounded by the number of processes. We refer the interested the reader to
second author’s previous work~\cite{DBLP:journals/pacmpl/GiustoFLL23}, and more generally to Courcelle's theorem~\cite{Courcelle10}, to better understand how this question is related to the decidability of satisfiability checking or to local, PTIME model-checking. Proofs omitted here can be found in Appendix~\ref{app:proofs}.
In this section, all histories are implicitly assumed to be \underline{finite}.

\subsection{Problem statement}\label{sec:graphs-finite-histories}
% A history of a program execution can be perceived as a timeline gathering different data updates. Key information held in a timeline includes when operations start and return, and the order in which they are performed. Since the \textsl{returns-before} relation holds exactly this information, it will serve as the foundation of our graphs. Graphs are data structures that relate some elements (referred to as \textsl{vertices} or \textsl{nodes}) according to some relation represented by \textsl{edges}. In this work, we focus on graphs that are restricted to edges having exactly two end-points. Since edges are intended to represent the \textsl{returns-before} relation, the \textsl{vertices} of these graphs correspond precisely to the operations in a history. 

We adopt the conventional notation for a graph $G=(V,E)$ with $V$ the set of vertices and $E$, where $E$ is a subset of $V\times V$, or equivalently a binary relation on $V$. The rb-graph $G_H$ associated with an history $H$ is $G_H=(H,\rb)$.

The rb-graph of an history is overly dense, with redundant information. 
This redundancy makes it impractical to define a bounded tree decomposition of such a graph. We therefore aim to define a new binary relation $\rightarrow$ on operations with the two following properties.
\begin{itemize}
	\item $\rb$ is the transitive closure of $\rightarrow$.
	As a consequence, the MSO theory of $G_H$ can be reduced to the one of $G=(H,\rightarrow)$.
	\item $(H,\rightarrow)$ admits a tree decomposition whose width is bounded by a number $n$ that only depends on the meta parameters (more precisely, it will only depend on $\mathbb{P}$)
\end{itemize}
The graph $(H,\rightarrow)$ is called a generator of $(H,\rb)$.
We will achieve something even stronger than what we just stated, 
as we will instead show that the generator $G=(H,\rightarrow)$ is of bounded \emph{cutwidth}. 

The cutwidth of an undirected graph $G$ is the minimum \mbox{$k \in \mathbb{N}$} such that for some linear ordering of the vertices of ${G}$ along a horizontal axis, every vertical cut intersects at most $k$ edges. 
Using the cutwidth is advantageous in our work, as vertices of graphs of history already represent a timeline, which is a linear ordering. 
To rigouroulsy define the cutwidth of a graph, it is necessary to introduce the notion of \textsl{cut} of a graph. 
Let $G=(V,E)$ be a graph with $n$ vertices $v_1,\ldots,v_n$, and $\sigma$ a permutation\footnote{in other words, a linear ordering/an enumeration of the vertices of $G$} of $\{1,\ldots,n\}$, and $\ell\in\{1,\ldots n\}$.
$$
Cut_\ell(\sigma,G)\eqdef (\{v_{\sigma(1)},...,v_{\sigma(\ell)}\},\{v_{\sigma(\ell+1)},...,v_{\sigma(n)}\})
$$
defines a \textsl{cut} of $G$. An edge $(u,v)$ crosses a cut $Cut_\ell(\sigma,G)=(L,R)$, if either $(u,v) \in L\times R$ or $(u,v)\in R\times L$. 
For $G=(V,E)$ a graph, let $\mathsf{Perm}(V)$ denote the set of all possible permutations of the vertices of $G$. 
$$
cutwidth(G) \eqdef \displaystyle \min_{\sigma\in \mathsf{Perm}(V)}\; \max_{1\leq\ell\leq n} \big|\big\{\;
(u,v)\;\big| \; (u,v) \mbox{ crosses } Cut_\ell(\sigma,G)
\;\big\}\big|.
$$

In the remainder, we first construct the generator graph $G$ of a given history, then we bound its cutwidth. As explained above, we will bound its cutwidth through a bound on $\max_{1\leq\ell\leq n} \big|\big\{\;
(u,v)\;\big| \; (u,v) \mbox{ crosses } Cut_\ell(\sigma,G)
\;\big\}\big|$ for the linear ordering $\sigma=ord(H)$ defined by the timestamps of the operations:
$$
ord(H) \eqdef (a_1,\ldots,a_n) \qquad\mbox{ such that for all } i<j, a_i.stime <a_j.stime.
$$

\subsection{Constructive definition of the subgraph}
We now define the relation $\rightarrow$ using an algorithm that iteratively adds new edges to an edge set $E$.  We write $\opa\patharrow\opb$ to mean that there exists a path from $\opa$ to $\opb$ in the graph $G=(V,E)$.  Remember that we aim at $\rb$ being the transitive closure of $\rightarrow$.

The set of all direct successors of $\opa$, written as $\succs(\opa)$, is defined as $\succs(\opa) = \{\opb \mid \opb \in H  \land \opa\dsucp{}\opb \} \subset H$ (see Definition~\ref{def:successor}). 
Observe that for all $\opa \in H$, we have $|\succs(\opa)| \leq |\mathbb{P}|$ because any operation has at most one direct successor on each process. 
The algorithm operates as follows. It takes as input $\textsl{ord}(H)$, where $H\in\historyset$, and traverses  $\textsl{ord}(H)$ in reverse, starting from the last element and proceeding to the first. For each operation, its direct successors are extracted, sorted according to \textsl{ord}, and then processed one by one. If a path to a given successor already exists in the graph, no action is taken; otherwise, a direct edge to that successor is added. This procedure ensures that the number of edges is minimized. The algorithm is formalized as follows.

\begin{algorithm}[H]
\renewcommand{\thealgorithm}{} % removes numbering
\caption{Building a graph of finite history}\label{alg:alg1}
\begin{multicols}{2}
\begin{algorithmic}[1]
\Require $H\in\historyset$ with $\textsl{ord}(H)=(\opa_1,...,\opa_n)$, \mbox{$n \in \mathbb{N}$}
\Ensure $G=(H,\rightarrow)$
\State $i \gets n-1$
\State $E\gets\emptyset$
\State $G\gets (H,E)$
\While{$i > 0$}
\State $S \gets (b_1,...,b_m)=\textsl{ord}(\succs(\opa_i))$
\For{$j\in\{1,...,m\}$}
\If{$\neg(\opa_i\patharrow\opb_m)$}
	\State $E \gets E \cup \{ (\opa_i,\opb_m) \}$
	\State $G \gets (H,E)$
\EndIf
\EndFor
\State $i \gets i-1$
\EndWhile
\end{algorithmic}
\end{multicols}
\end{algorithm}

\begin{remark} 
    When building a directed graph as above, each directed edge can be interpreted as a \say{jump forward in time} —-- that is, traversing an edge corresponds to moving along the timeline of certain operations in the history.
\end{remark}

We denote by $\algone(H)$ the graph of the finite history $H$ computed by this algorithm. Let $\algoneset_m,m\in\mathbb{N}$ denote the set of all such graphs for histories in  $\historyset$ with $|\mathbb{P}|= m $. It is the finite execution graph set; for $m\in\mathbb{N}$,
\[
\algoneset_m \eqdef \big\{\;  \algone(H) \;\mid \;H \in \historyset \text{ and }|\mathbb{P}|= m  \;\big\}
\]

\subsection{Bounding the cutwidth of histories}
In this section, we prove that $\algoneset$ has uniformly bounded cutwidth, which, as a by-product, establishes the correctness of the algorithm.

\subsubsection{Preliminary results}

It is worth noting that both $\rb$ and $\patharrow$ are transitive. Indeed, given $\opa \rb \opb$ and $\opb \rb \opc$,   it follows naturally that $\opa\rb\opc$. Similarly, for  $\patharrow$, if $\opd \patharrow \ope$ and $\ope \patharrow \opf$, then $\opd \patharrow \opf$.

For $\opa$ a vertex of a graph, we denote its out-degree as $deg_{out}(\opa)$, and its in-degree as $deg_{in}(\opa)$. The \textsl{out-degree} of a vertex is the number of edges connected to it that are directed away from the vertex. The \textsl{in-degree} is the number of connected edges directed towards the vertex.

\begin{lemma}\label{lemma:degout}
Let $G=(H,\rightarrow)\in\algoneset_m,m\in\mathbb{N}$. Then, for all $\opa\in H$, \mbox{$deg_{out}(\opa) \leq m$}.
\end{lemma}

\begin{proposition}\label{prop:finishesBeforeCollection}
Let $H\in\historyset$ be a history. For $\opa,\opb \in H$,
\begin{center}
$\opa\rb\opb \Leftrightarrow $ there exists a tuple $(\oph_1,...,\oph_q),\; \oph_i \in H, q\in\mathbb{N}$ such that 
$\opa \dsucp{p} \oph_1 \dsucp{p} \hdots \dsucp{p} \oph_q \dsucp{p} \opb$ with  $p:=\opb.proc$.
\end{center}
\end{proposition}
%%%%%%%%%%%%%%%%%%%%%%%%%%%%%%%%%%%%%%%%%%%%%%%%%%%%%%

\begin{proposition}\label{prop:looparrowPATH}
Let $(H,\rightarrow)\in\algoneset_m,m\in\mathbb{N}$. For \mbox{$\opa,\opb \in H$}, $\opa\rb\opb \Leftrightarrow \opa\patharrow\opb$.
\end{proposition}
We have established that for a history $H\in\historyset$, each vertex of $G=\algone(H)$ has at most $m=|\mathbb{P}|$ outgoing edges. We now turn to the complementary property: showing that each vertex has at most $m$ incoming edges.
\begin{lemma}\label{lemma:degin}
Let $G=(H,\rightarrow)\in\algoneset_m,m\in\mathbb{N}$. Then, for all $\opa\in H$, \mbox{$deg_{in}(\opa) \leq m$}.
\end{lemma}

\subsubsection{Main proof}

To prove cutwidth boundedness of a graph $G$, we need to define two specific sets of nodes  based on the definition of \textsl{cut}. These sets are as follows: the set containing the last operations launched on each process before the cut $\ell$ is denoted by $\Gamma_\ell(G)$, and the set containing the first operations launched on each process after the cut $\ell$ is denoted by $\Lambda_\ell(G)$. 
% In order to formally define $\Gamma_\ell(G)$ and $\Lambda_\ell(G)$, we first need to introduce the sets $\gamma_p(L)$ and $\lambda_p(L)$: the singleton $\gamma_p(L)$ contains the last operation of $L$ to be launched on process $p$, and the singleton $\lambda_p(R)$ contains the first operation of $R$ to be launched on process $p$.

Let $H\in\historyset$ be a history. For a set $O\subseteq H$, and a process $p$, 
we write $\gamma_p(O)$ to denote the singleton containing the last operation of $O$ to be launched on process $p$, and $\lambda_p(O)$ to denote the singleton containing the first operation of $O$ to be launched on process $p$. Formally,
$$
\begin{array}{ccl}
    \gamma_p(O) & \eqdef & \{a\in O\mid a.proc=p \mbox{ and }\forall \opb\in O.\;\opb\ssrel\opa\Rightarrow \opb.stime < \opa.stime \} \\
    \lambda_p(O) & \eqdef & \{a\in O\mid  a.proc=p \mbox{ and }\forall \opb\in O.\;\opb\ssrel\opa\Rightarrow \opa.stime < \opb.stime \}
\end{array}
$$

Let us fix $1\leq\ell\leq n$ and let $(L,R)=\textsl{Cut}_\ell(\textsl{ord},G)$.
The sets \mbox{$\Gamma_\ell(G)$} and \mbox{$\Lambda_\ell(G)$} are now defined as\footnote{
Observe that it is possible to have \mbox{$\gamma_p(L)=\emptyset$} and/or \mbox{$\lambda_p(R)=\emptyset$}.}
\begin{gather*}
\Gamma_\ell(G) := \displaystyle\bigcup_{p \in \mathbb{P}}\gamma_p(L)
\quad\text{and}\quad
\Lambda_\ell(G) := \displaystyle\bigcup_{p \in \mathbb{P}}\lambda_p(R)
\end{gather*}

\noindent By definition $|\Gamma_\ell(G)|,\;|\Lambda_\ell(G)|\leq|\mathbb{P}|$. We now introduce notation for the sets of remaining nodes not included in $\Gamma_\ell$ or $\Lambda_\ell$: $L_\ell := L \setminus \Gamma_\ell $ and $R_\ell := R \setminus \Lambda_\ell$.

\begin{lemma}\label{prop:crossRedges}
Let a history $H\in\historyset$ and $G=(H,\rightarrow)\in\algone(H)$. For a cut 
\mbox{$(L,R)=\textsl{Cut}_\ell(\textsl{ord},G)$}, with \mbox{$1\leq\ell\leq|H|$}, there is no edge that crosses the cut $\ell$ from $R$ to $L$: $\nexists\; \opa,\opb\in H:\; \opa\rightarrow\opb\;\land\;\opa\in R\;\land\;\opb\in L$.
\end{lemma}

Let $G\in\algoneset_m,m\in\mathbb{N}$, and let \mbox{$(L,R)=\textsl{Cut}_\ell(\textsl{ord},G)$} be a cut of $G$.
By Lemma~\ref{prop:crossRedges} no edge originating from $R=\Lambda_\ell \cup R_\ell$ can cross the cut. Edges in $L_\ell\times R_\ell$ cannot crosse the cut either. However, edges in $L_\ell\times\Lambda_\ell$, $\Gamma_\ell\times\Lambda_\ell$ and $\Gamma_\ell\times R_\ell$ can cross the cut $\ell$. 

We now aim to prove that no edge in $L_\ell\times R_\ell$ crosses the cut $\ell$.

\begin{lemma}\label{lemma:noCrossingLlRl}
Let a history $H\in\historyset$ and $G=(H,\rightarrow)\in\algone(H)$. For a cut 
\mbox{$(L,R)=\textsl{Cut}_\ell(\textsl{ord},G)$}, with \mbox{$1\leq\ell\leq|H|$}, there is no edge that crosses the cut $\ell$ from $L_\ell(G)$ to $R_\ell(G)$: $\nexists\; \opa,\opb\in H:\; \opa\rightarrow\opb\;\land\;\opa\in L_\ell(G)\;\land\;\opb\in R_\ell(G)$.
\end{lemma}

\begin{theorem}\label{thm:cutwidthBounded}
Let \mbox{$G\in\algoneset_m,m\in\mathbb{N}$}. Then, $cutwidth(G)\leq 2m^2$.
\end{theorem}
\begin{proof}
Let a cut \mbox{$(L,R)=\textsl{Cut}_\ell(\textsl{ord},G)$}, with$1\leq\ell\leq|H|$. Let $\Gamma_\ell\eqdef \Gamma_\ell(G)$, $\Lambda_\ell\eqdef\Lambda_\ell(G)$, $L_\ell\eqdef L_\ell(G)$, $R_\ell\eqdef R_\ell(G)$. Cutwidth is traditionally defined for undirected graphs. However, by Lemma~\ref{prop:crossRedges}, only edges originating in $L$ can cross the cut at level $\ell$. Intuitively, this means that edges cross the cut only \say{from left to right}. For this reason, we retain the directed nature of the graph to maintain clarity throughout the proof. Thus, the core idea of the proof is to bound the number of edges that can \say{land} on the right-hand side of the cut; that is, the number of edges originating from $L_\ell\cup\Gamma_\ell$ and landing in $\Lambda_\ell\cup R_\ell$. 
We have $|\Lambda_\ell| \leq m$, and by Lemma~\ref{lemma:degin} each element of $\Lambda_\ell$ has at most $m$ incoming edges. Therefore there are at most $m^2$ edges \say{landing} in $\Lambda_\ell$.
By Lemma~\ref{lemma:noCrossingLlRl}, there are no edges in $L_\ell\times R_\ell$, which implies that each edge crossing the cut $\ell$ and landing in $R_\ell$ must originate from $\Gamma_\ell$. Since $|\Gamma_\ell| \leq m$ and Lemma~\ref{lemma:degout} states that each element of $\Gamma_\ell$ has at most $m$ outgoing edges, the number of edges in $\Gamma_\ell\times R_\ell$ is also bounded by $m^2$. This implies that at most $m^2$ edges cross the cut and land in $R_\ell$.
Adding the two bounds, we obtain a total of at most $m^2+m^2=2m^2$ crossing edges, which establishes the result.

\end{proof}

\section{Conclusion and Future Work}
We introduced HistMSO, a monadic second order logic over histories and abstract executions, the prominent execution model for replicated data systems. 
We showed that the logic can express 39 out of 42 consistency models inventoried by Viotti and Vukolic~\cite{ViottiVukolic}. 
We proposed a translation of HistMSO to MSO over words, leveraging MONA as a tool for automated reasoning on consistency models.
We established the soundness and completeness of this translation for all finite and $\omega$-infinite histories, assuming that all the meta-parameters (sets of processes, of values, and of replicated objects) were finite.
We extended this translation to abstract executions under two assumptions: real-time arbitration, and $k$-transient visibility.
From this translation, we derive the decidability of HistMSO on finite (respectively $\omega$-infinite) histories and on real-time, $k$-transient abstract executions.
As a direct application, we could fully automate the verification of the hierarchy of consistency models proposed by Viotti and Vukolic.
We also established a bound on the cutwidth complexity of the interval graph associated with a given history, grounding connections with previous work on bounded the tree width of some communication models under bounded asynchrony~\cite{DBLP:journals/pacmpl/GiustoFLL23}.

We believe that the main strength of our framework is to support reasoning about a broad class of consistency properties within a unified logical setting.
The other strength of our work is to introduce an encoding of execution traces suitable for consistency verification using finite state automata under minimal assumptions about the structure of these traces.
Beyond automated reasoning (i.e. deciding satisfiability of a HistMSO theory), our framework also allows to monitor~\cite{havelund2018runtime} an execution trace of an instrumented replicated data system, and possibly catch consistency violations of a system.
This monitoring is however \emph{offline} (the log trace must first be generated by the instrumented system, then model-checked), and \emph{centralised} (the model-checker needs to access the whole trace, not just the one of each process).
An interesting question, that has been successfully addressed in other automata-based framework, is whether online/runtime, and/or distributed monitoring would also be possible. We leave this question for future work, and for more immediate future work, we aim at benchmarking our approach with a tool implementation.

% Knowing that a trace can be represented by an automaton is valuable for several reasons. Since our framework is restricted to finite processes and histories, the execution of a program over these processes can be decomposed into a separate automaton per process. Each process is therefore associated with a finite automaton that captures the specific code fragments it executes. To verify that any history resulting from the parallel execution of these automata satisfies a given consistency model, we proceed in two steps: first we define a class of histories that captures all possible outcomes of the program; and then verify whether each of these histories satisfies the specified consistency property. These constraints can be formalized as intersections of finite automata. Since each finite automaton corresponds to a regular language, the problem reduces to checking whether the intersection language is empty. If the intersection is empty, then no execution of the program satisfies the consistency model.

% However, our work is currently restricted to finite domains. An immediate extension would be to investigate the decidability of MSO-definable properties over infinite traces. The current decidability result relies on an algorithmic construction, which is only applicable to finite traces. Nonetheless, if the underlying graph structure of execution traces can be expressed purely in logical terms --- using unions, intersections, and negations of relations, rather than an explicit algorithm --- then it is possible to extend the decidability result to infinite traces.

\begin{credits}
\subsubsection{\ackname} 
% Ce travail a bénéficié d'une aide du gouvernement français, gérée par l'Agence Nationale de la Recherche au titre du Pl­an d’investissement France 2030, dans le cadre du projet "UCA DS4H" portant la référence n° ANR-17-EURE-0004
This work was supported by a grant from the French government under
the "Investissements d’Avenir" program managed by the National
Research Agency (ANR), reference ANR-17-EURE-0004, and by
the project UCA DS4H.

\subsubsection{\discintname}
%It is now necessary to declare any competing interests or to specifically state that the authors have no competing interests. Please place the statement with a bold run-in heading in small font size beneath the (optional) acknowledgments\footnote{If EquinOCS, our proceedings submission system, is used, then the disclaimer can be provided directly in the system.}, for example: The authors have no competing interests to declare that are relevant to the content of this article. Or: Author A has received research grants from Company W. Author B has received a speaker honorarium from Company X and owns stock in Company Y. Author C is a member of committee Z.
The authors have no competing interests to declare that are relevant to the content of this article.
\end{credits}

\bibliographystyle{splncs04}
\bibliography{references}

@inproceedings{DBLP:conf/fm/ChevrouH0Q19,
  author       = {Florent Chevrou and
                  Aur{\'{e}}lie Hurault and
                  Shin Nakajima and
                  Philippe Qu{\'{e}}innec},
  editor       = {Emil Sekerinski et al},
  title        = {A Map of Asynchronous Communication Models},
  booktitle    = {Formal Methods. {FM} 2019 International Workshops - Porto, Portugal,
                  October 7-11, 2019, Revised Selected Papers, Part {II}},
  series       = {Lecture Notes in Computer Science},
  volume       = {12233},
  pages        = {307--322},
  publisher    = {Springer},
  year         = {2019},
  url          = {https://doi.org/10.1007/978-3-030-54997-8\_20},
  doi          = {10.1007/978-3-030-54997-8\_20},
  bibsource    = {dblp computer science bibliography, https://dblp.org}
}

@article{DBLP:journals/pacmpl/GiustoFLL23,
	author = {Cinzia {Di Giusto} and Davide Ferr{\'{e}} and Laetitia Laversa and {\'{E}}tienne Lozes},
	bibsource = {dblp computer science bibliography, https://dblp.org},
	doi = {10.1145/3571248},
	journal = {Proc. {ACM} Program. Lang.},
	number = {{POPL}},
	pages = {1601--1627},
	title = {A Partial Order View of Message-Passing Communication Models},
	url = {https://doi.org/10.1145/3571248},
	volume = {7},
	year = {2023}}

@article{ENGELS2002253,
	author = {A.G. Engels and S. Mauw and M.A. Reniers},
	doi = {https://doi.org/10.1016/S0167-6423(02)00022-9},
	issn = {0167-6423},
	journal = {Science of Computer Programming},
	number = {3},
	pages = {253-292},
	title = {A hierarchy of communication models for Message Sequence Charts},
	url = {https://www.sciencedirect.com/science/article/pii/S0167642302000229},
	volume = {44},
	year = {2002}}

@inproceedings{Courcelle10,
	Author = {Bruno Courcelle},
	Booktitle = {FSTTCS},
	Ee = {http://dx.doi.org/10.4230/LIPIcs.FSTTCS.2010.13},
	Pages = {13-29},
	Series = {LIPIcs},
	Title = {Special tree-width and the verification of monadic second-order graph properties},
	Volume = {8},
	Year = {2010},
	address         = {Chennai, India},
	publisher       = {Schloss Dagstuhl - Leibniz-Zentrum f{\"{u}}r Informatik},
}

@misc{ViottiVukolic,
  author = {Paolo Viotti and Marko Vukolic},
  title = {Consistency in Non-Transactional Distributed Storage Systems.},
  year = {2016},
  howpublished = {\url{http://vukolic.com/consistency-survey.pdf}}
}

@misc{PrinciplesOfEventualConsistency,
  author = {Sebastian Burckhardt},
  title = {Principles of Eventual Consistency},
  year = {2014},
  howpublished = {\url{https://www.microsoft.com/en-us/research/wp-content/uploads/2016/02/final-printversion-10-5-14.pdf}}
}

@article{DBLP:journals/corr/abs-2510-21304,
  author       = {Hagit Attiya and
                  Constantin Enea and
                  Enrique Rom{\'{a}}n{-}Calvo},
  title        = {Arbitration-Free Consistency is Available (and Vice Versa)},
  journal      = {CoRR},
  volume       = {abs/2510.21304},
  year         = {2025},
  url          = {https://doi.org/10.48550/arXiv.2510.21304},
  doi          = {10.48550/ARXIV.2510.21304},
  eprinttype    = {arXiv},
  eprint       = {2510.21304},
  bibsource    = {dblp computer science bibliography, https://dblp.org}
}

@InProceedings{demri-nowak-05-transfinite-words,
author="Demri, St{\'e}phane
and Nowak, David",
editor="Peled, Doron A.
and Tsay, Yih-Kuen",
title="Reasoning About Transfinite Sequences",
booktitle="Automated Technology for Verification and Analysis",
year="2005",
publisher="Springer Berlin Heidelberg",
address="Berlin, Heidelberg",
pages="248--262",
isbn="978-3-540-31969-6"
}

@inproceedings{DBLP:conf/spaa/SinglaRH97,
  author       = {Aman Singla and
                  Umakishore Ramachandran and
                  Jessica K. Hodgins},
  title        = {Temporal Notions of Synchronization and Consistency in Beehive},
  booktitle    = {Proceedings of the Ninth Annual ACM Symposium on Parallel Algorithms and Architectures,
                  SPAA 1997, Santa Barbara, California, USA, July 23-25, 1997},
  pages        = {211--220},
  publisher    = {ACM},
  year         = {1997},
  url          = {https://doi.org/10.1145/258492.258513},
  doi          = {10.1145/258492.258513},
}

@article{HerlihyW90,
  author       = {Maurice Herlihy and
                  Jeannette M. Wing},
  title        = {Linearizability: A Correctness Condition for Concurrent Objects},
  journal      = {ACM Transactions on Programming Languages and Systems (TOPLAS)},
  volume       = {12},
  number       = {3},
  pages        = {463--492},
  year         = {1990},
  publisher    = {ACM},
  url          = {https://doi.org/10.1145/78969.78972},
  doi          = {10.1145/78969.78972},
}

@misc{MONA,
title={MONA},
author={Tobias Nipkow and Andreas Weber},
  url= {https://www.brics.dk/mona/index.html},
  year={2020}
}

@misc{MONAguide,
title={MONA Version 1.4 User Manual},
author={Nils Klarlund and Anders Møller},
url={https://www.brics.dk/mona/mona14.pdf},
year={2001}
}

@misc{Burckhardt13,
  author = {Sebastian Burckhardt, Alexey Gotsman, Hongseok Yang},
  howpublished = {Technical Report MSR-TR-2013-39},
  title = {Understanding eventual consistency},
  year = {2013}
}

@article{alur1994theory,
  title={A theory of timed automata},
  author={Alur, Rajeev and Dill, David L},
  journal={Theoretical computer science},
  volume={126},
  number={2},
  pages={183--235},
  year={1994},
  publisher={Elsevier}
}

@inproceedings{bengtsson1995uppaal,
  title={UPPAAL—a tool suite for automatic verification of real-time systems},
  author={Bengtsson, Johan and Larsen, Kim and Larsson, Fredrik and Pettersson, Paul and Yi, Wang},
  booktitle={International hybrid systems workshop},
  pages={232--243},
  year={1995},
  organization={Springer}
}

@inproceedings{bozga1998kronos,
  title={Kronos: A model-checking tool for real-time systems},
  author={Bozga, Marius and Daws, Conrado and Maler, Oded and Olivero, Alfredo and Tripakis, Stavros and Yovine, Sergio},
  booktitle={International Conference on Computer Aided Verification},
  pages={546--550},
  year={1998},
  organization={Springer}
}

@inproceedings{kovacs2013first,
  title={First-order theorem proving and Vampire},
  author={Kov{\'a}cs, Laura and Voronkov, Andrei},
  booktitle={International Conference on Computer Aided Verification},
  pages={1--35},
  year={2013},
  organization={Springer}
}

@inproceedings{de2008z3,
  title={Z3: An efficient SMT solver},
  author={De Moura, Leonardo and Bj{\o}rner, Nikolaj},
  booktitle={International conference on Tools and Algorithms for the Construction and Analysis of Systems},
  pages={337--340},
  year={2008},
  organization={Springer}
}

@article{holzmann2007design,
  title={The design of a multicore extension of the SPIN model checker},
  author={Holzmann, Gerard J and Bosnacki, Dragan},
  journal={IEEE Transactions on Software Engineering},
  volume={33},
  number={10},
  pages={659--674},
  year={2007},
  publisher={IEEE}
}

@inproceedings{yu1999model,
  title={Model checking TLA+ specifications},
  author={Yu, Yuan and Manolios, Panagiotis and Lamport, Leslie},
  booktitle={Advanced research working conference on correct hardware design and verification methods},
  pages={54--66},
  year={1999},
  organization={Springer}
}

@inproceedings{weidenbach2009spass,
  title={SPASS Version 3.5},
  author={Weidenbach, Christoph and Dimova, Dilyana and Fietzke, Arnaud and Kumar, Rohit and Suda, Martin and Wischnewski, Patrick},
  booktitle={International Conference on Automated Deduction},
  pages={140--145},
  year={2009},
  organization={Springer}
}

@inproceedings{havelund2018runtime,
  title={Runtime verification-17 years later},
  author={Havelund, Klaus and Ro{\c{s}}u, Grigore},
  booktitle={International Conference on Runtime Verification},
  pages={3--17},
  year={2018},
  organization={Springer}
}

@InCollection{KlaNieSun:casestudyautver,
  author = 	 {Klarlund, N. and Nielsen, M. and Sunesen, K.},
  title = 	 {A case study in automated verification
		based on trace abstractions},
  booktitle = 	 {Formal System Specification,
       		 The RPC-Memory Specification Case Study},
  publisher =	 {Springer Verlag},
  year =	 1996,
  editor =	 {Broy, M. and Merz, S. and Spies, K.},
  volume =	 1169,
  series =	 {LNCS},
  pages  = {341-374},
}

@INPROCEEDINGS{JenJoeKlaSch:AutVerPoinProgMonSecOrdLog,
	AUTHOR = {Jacob L. Jensen and Michael E. Joergensen and Nils Klarlund and Michael
I. Schwartzbach},
	TITLE = {Automatic Verification of Pointer Programs using Monadic Second-order Logic},
	BOOKTITLE = {PLDI '97},
	YEAR = {1997},
}

@Article{BasKla:BeyondFiniteHardware,
  author = 	 {Basin, D. and Klarlund, N.},
  title = 	 {Automata Based Symbolic Reasoning in
Hardware Verification},
  journal = 	 {Formal Methods In System Design},
  year = 	 1998,
  volume =	 13,
  pages =	 {255-288},
  note =	 {Extended version of: 
        ``Hardware verification using monadic 
              second-order logic," {\em CAV '95}, LNCS 939}}

@InProceedings{YakYak99,
  author = {Niels Damgaard and Nils Klarlund and Michael I. Schwartzbach},
  title = {{YakYak}: Parsing with Logical Side Constraints},
  booktitle = {Proceedings of DLT'99},
  year = {1999},
}

\appendix
\section{Additional Consistency Models}\label{app:more-consistency-models}
\subsubsection{Monotonic reads} states that \say{if a process has read a certain value $v$ from an object, any successive read operation
 will not return any value written before $v$}\cite{ViottiVukolic}. Formally, it is defined as follows.

$$
\begin{array}{rcl}
    a\xrightarrow{sorr} b & \eqdef & 
 a.type=read\land b.type=read \wedge a\sessionorder b \\
 \textsc{MonotonicReads} & \eqdef & 
 \forall a. \forall b.\forall c.  
(a\xrightarrow{vis}b\land b\xrightarrow{sorr} c)\Rightarrow a\xrightarrow{vis}c
\end{array}
$$

\subsubsection{Read your writes} ensures that if a \textsl{write} and then a \textsl{read} are performed by the same process, the \textsl{write} will be visible to the \textsl{read}.
\begin{gather*}
\textsc{ReadYourWrites} := \forall	a.\forall b.\Big( a.type=write\land b.type=read\Rightarrow ( a\xrightarrow{so}b \Rightarrow a\xrightarrow{vis}b  ) \Big)
\end{gather*}

 % \textsl{Eventual consistency} guarantess that all correct operations are eventually applied to all correct replicas, and that all operations return. Moreover, all correct replicas that had the same write operations performed will eventually reach an \say{equivalent state}. An operation or a replica is \textsl{correct} if it does not fail. To define eventual consistency, the \textsc{NoCircularCausality} consistency model must be defined. In \cite{ViottiVukolic}, it is defined as \mbox{$acyclic(hb)$}, where $hb := (\xrightarrow{so} \cup \vis)^+$ represents the \textsl{happens-before} relation, and $\xrightarrow{so}$ denotes the \textsl{session order} relation. Specifically, we have $\opa \xrightarrow{so} \opb$ if and only if $\opa \ssrel \opb \land \opa \rb \opb$. The relation $hb$ is MSO-expressible, as it is a transitive closure. Consequently, $acyclic(hb)$ is also MSO-expressible, since it is a restriction of $hb$. \textsl{Eventual consistency} is defined in \cite{ViottiVukolic} as the logical conjunction \mbox{$\textsc{EventualVisibility}\land\textsc{NoCircularCausality}\land\textsc{RVal}$}. However, since $\textsc{EventualVisibility}$ is not relevant in our framework, it is defined as follows.

% \begin{center}
% $\textsc{EventualConsistency} := \textsc{NoCircularCausality}\land\textsc{RVal}$
% \end{center}

% EXEMPLE: ou pas eventual consistency (autre nom??) + expliquer no circular causality

\section{Omitted Proofs}\label{app:proofs}
\setcounter{lemma}{0}
\setcounter{proposition}{0}
\begin{lemma}
Let $G=(H,\rightarrow)\in\algoneset_m,m\in\mathbb{N}$. Then, for all $\opa\in H$, \mbox{$deg_{out}(\opa) \leq m$}.
\end{lemma}
\begin{proof}
In the algorithm, line 5 marks the beginning of the iteration over all operations in $\textsl{ord}(H)$.
Let $\opa$ denote the operation currently being processed in the algorithm. Between lines 7 and 12, the algorithm may add up to $m$ out-going edges from $\opa$, since by definition $\sigma(\opa)\leq |\mathbb{P}|=m$.

As each vertex has its out-going edges built as above, all vertices of $G$ have an out-going degree of at most $m$.

\end{proof}

%%%%%%%%%%%%%%%%%%%%%%%%%%%%%%%%%

\begin{proposition}
Let $H\in\historyset$ be a history. For $\opa,\opb \in H$,

\begin{center}
$\opa\rb\opb \Leftrightarrow $ there exists a tuple $(\oph_1,...,\oph_q),\; \oph_i \in H, q\in\mathbb{N}$ such that 
$\opa \dsucp{p} \oph_1 \dsucp{p} \hdots \dsucp{p} \oph_q \dsucp{p} \opb$ with  $p:=\opb.proc$.
\end{center}
\end{proposition}

\begin{proof}\slshape
$\Rightarrow)$

Let $\opa,\opb \in H$ with $\opa\rb\opb$, and let $T$ denote the tuple we aim to construct. We distinguish two main cases.
\begin{itemize}
	\item[$\ast$] if $\opb$ is a direct successor of $\opa$; that is, there exists $p \in \mathbb{P}$ such that $\opa\dsucp{p}\opb$, then we set $T=()$
	\item[$\ast$]  if $\opb$ is not a direct successor of $\opa$, we proceed as follows.
\end{itemize}

Let $p=\opb.proc$. Having $\opa\rb\opb$ implies that on process $p$ there are none or some operations executed after $\opa$ returns and before $\opb$ starts. We build $T$ by collecting all of those successive operations, starting with $\oph_1$ such that \mbox{$\opa\dsucp{p}\oph_1$}, and then $\oph_2,...,\oph_q$ with 
\mbox{$\opa\dsucp{p} \oph_1 \dsucp{p} \oph_2 \dsucp{p} ... \dsucp{p} \oph_q \dsucp{p} \opb$}. We set $T=(\oph_1,\oph_2,...,\oph_q)$.

$\Leftarrow)$

Let $\opa,\oph_1,...,\oph_q,\opb \in H$ with 
\mbox{$\opa \dsucp{p} \oph_1 \dsucp{p} ... \dsucp{p} \oph_q \dsucp{p} \opb$} and  $p:=\opb.proc$. By definition of $\dsucp{p}$, we have for all \mbox{$1 \leq i \leq q$}, $\oph_i.rtime<\oph_{i+1}.stime$ as well as \mbox{$\opa.rtime<h_1.stime$} and \mbox{$\oph_q.rtime<\opb.stime$}. We recall that for all \mbox{$1 \leq i \leq q$}, $h_i.stime<h_i.rtime$. By the transitivity of $<$, we have $\opa.rtime<\opb.stime$ which means $\opa \rb\opb$ by definition.

\end{proof}

%%%%%%%%%%%%%%%%%%%%%%%%%%%%%%%%%%%%%%%%%%%%%%%%
\begin{proposition}
Let $G=(H,\rightarrow)\in\algoneset_m,m\in\mathbb{N}$. For \mbox{$\opa,\opb \in H$},
\begin{center}
$\opa\rb\opb \Leftrightarrow \opa\patharrow\opb$.
\end{center}
\end{proposition}
\begin{proof}\slshape
$\Rightarrow)$

Let $G=(H,\rightarrow)\in\algoneset_m$, and $\opa,\opb \in H$ such that $\opa\rb\opb$. By Proposition~\ref{prop:finishesBeforeCollection}, there exists a tuple \mbox{$T=(\oph_1,...,\oph_q),\; \oph_i \in H, q\in\mathbb{N}$} such that 
\begin{center}
\mbox{$\opa\dsucp{p} \oph \dsucp{p} \oph_2 \dsucp{p} ... \dsucp{p} \oph_q \dsucp{p} \opb$}

\end{center}
with  $p=\opb.proc$. 

Lines 7 and 8 of the algorithm ensure that a path always exists from any operation to its direct successor on a given process, if such a successor exists. And since $\oph_i \dsucp{p} \oph_{i+1}$  is defined to mean that  $\oph_{i+1}\in\sigma(\oph_i)$, we have

\begin{center}
$\opa \patharrow \oph_1 \patharrow \oph_2 \patharrow ... \patharrow \oph_{q-1} \patharrow \oph_q \patharrow \opb$.
\end{center}

Thus, $\opa\patharrow\opb$ by transitivity of $\patharrow$.

$\Leftarrow)$

We have $\opa\patharrow\opb$. Let the $\oph_1,...,\oph_q$ be the vertices on the path from $\opa$ to $\opb$ such that

\begin{center}
$
\opa \rightarrow \oph_1 \rightarrow \oph_2 \rightarrow ... \rightarrow \oph_{q-1} \rightarrow \oph_q \rightarrow \opb
$
\end{center}

%By Observation~\ref{obs:directPREDandSUCC}, with  $p=\opb.proc$,

By construction, $\forall \opc,\opd\in H:\opc\rightarrow\opd\Rightarrow \opc\dsucp{\opd.proc}\opd$. Thus, for $p=\opb.proc$,

\begin{center}
$
\opa \rightarrow \oph_1 \rightarrow \oph_2 \rightarrow ... \rightarrow \oph_{q-1} \rightarrow \oph_q \rightarrow \opb
 \quad\Rightarrow\quad$
\mbox{$\opa \dsucp{p} \oph_1 \dsucp{p} ... \dsucp{p} \oph_q \dsucp{p} \opb$} 
\end{center}

$\Rightarrow \opa \rb \opb$ by Proposition~\ref{prop:finishesBeforeCollection}.
\end{proof}

%%%%%%%%%%%%%%%%%%%%%%%%%%%%%%%%%%%%

\begin{lemma}
Let $G=(H,\rightarrow)\in\algoneset_m,m\in\mathbb{N}$. Then, for all $\opa\in H$, \mbox{$deg_{in}(\opa) \leq m$}.
\end{lemma}
\begin{proof} 

We proceed by showing that each vertex has at most one incoming edge from each process. Formally, we aim to prove that for all $\opc \in H$, no $\opa,\opb\in H$ exist such that $\opa\ssrel\opb\land\opa\rightarrow\opc \land \opb\rightarrow\opc$~. We proceed by contradiction. Suppose that there exists $\opa,\opb\in H$ such that $\opa\rightarrow\opc\land\opb\rightarrow\opc\land\opa\ssrel\opb$. As $\opa$ and $\opb$ are executed on the same process, we either have $\opa \rb \opb$ or $\opb \rb \opa$. Let $\opa \rb \opb$ without loss of generality. We also have $\opa \rb \opc$ and $\opb \rb \opc$ as $\opc$ starts after both $\opa$ and $\opb$ finish by definition of $\rightarrow$. Our algorithm traverses $\textsl{ord}(H)$ in reverse order; meaning that our operations will be treated in the following order: first $\opc$, then $\opb$ and finally $\opa$. Between each, other operations may be processed, though not impacting our reasoning, so we will ignore them. Let $p$ be the process on which $\opa$ and $\opb$ are executed. Here is how our algorithm will proceed:
\begin{itemize}
	\item[$\ast$] Processing operation $\opc$: adding the needed out-going edges (will not matter in our proof as those edges are towards operations that occured after $\opc$ finished, and we are focused of $\opa$ and $\opb$ that occur before $\opc$)
	\item[$\ast$] Processing operation $\opb$: for the moment there are no out-going edges from $\opb$. We supposed the existence of the edge $\opb\rightarrow\opc$ which we therefore add to the graph, and which implies $\opb\rb\opc$
	\item[$\ast$] Processing operation $\opa$: as we supposed $\opa\rightarrow\opc$ to be at some point an edge of the graph we have by construction that $\opc$ is a direct successor of $\opa$. Let $\opf \in H$ be the operation such that $\opa \dsucp{p} \opf$; then either $\opf\rb\opb$ or $\opf=\opb$. In either case, the proof holds. Upon the completion of the processing of operation $\opa$, we have $\opa \patharrow \opf$ as well as $\opf \patharrow \opb$. Since $\opb\rightarrow\opc$ implies $\opb \patharrow \opc$, by transitivity of $\patharrow$ we have $\opa \patharrow \opc$. Since there already exists a path from $\opa$ to $\opc$, the edge $\opa\rightarrow\opc$ will not be added during the processing of $\opa$.
	
Hence the contradiction.
\end{itemize}
We have thus proven that each vertex has at most one in-coming edge from each process. Consequently, across all processes, each vertex can have at most $m$ in-coming edges.
\end{proof}

%%%%%%%%%%%%%%%%%%%%%%%%%%%%%%%%%%%%

\begin{lemma}
Let a history $H\in\historyset$ and a graph \mbox{$G=(H,\rightarrow)\in\algone(H)$}. For a cut 
\mbox{$(L,R)=\textsl{Cut}_\ell(\textsl{ord},G)$}, with \mbox{$1\leq\ell\leq|H|$}, there is no edge that crosses the cut $\ell$ from $R$ to $L$. That is, $\nexists\; \opa,\opb\in H:\; \opa\rightarrow\opb\;\land\;\opa\in R\;\land\;\opb\in L$.
\end{lemma}

\begin{proof}
We proceed by contradiction. Suppose there is such an edge $ \opa\rightarrow\opb$ with $\opa\in R$ and $\opb\in L$. The algorithm is designed so that any edge $\opc\rightarrow\opd$ implies $\opd\in\sigma(\opc)$, and therefore $\opc$ must return before $\opd$ starts. In particular, $\opa\rightarrow\opb$ implies $\opa.rtime<\opb.stime$.
However, by the definition of the cut $Cut_\ell(\textsl{ord},G)$, since $\opa\in R$ and $\opb\in L$, we also have \mbox{$\opb.stime\leq\opa.stime$}. Combining these inequalities, we obtain \mbox{$\opa.rtime<\opb.stime\leq\opa.stime$} which is a contradiction, as it violates the property that $\opa.stime<\opa.rtime$. Hence, such an edge cannot exist.
\end{proof}

%%%%%%%%%%%%%%%%%%%%%%%%%%%%%%%%%%%%

\begin{lemma}
Let a history $H\in\historyset$ and the graph \mbox{$G=(H,\rightarrow)\in\algone(H)$}. For a cut 
\mbox{$(L,R)=\textsl{Cut}_\ell(\textsl{ord},G)$}, with \mbox{$1\leq\ell\leq|H|$}, there is no edge that crosses the cut $\ell$ from $L_\ell(G)$ to $R_\ell(G)$. That is, \mbox{$\nexists\; \opa,\opb\in H:\; \opa\rightarrow\opb\;\land\;\opa\in L_\ell(G)\;\land\;\opb\in R_\ell(G)$}.
\end{lemma}
\begin{proof}
We proceed by contradiction. Let suppose the existence of $\opa\in L_\ell(G),\opb\in R_\ell(G)$ such that $\opa\rightarrow	\opb$, and let $p=\opa.proc$ and $q=\opb.proc$. 
As $\opa\notin \Gamma_\ell(G)$, there exists a vertex $\opc\in \Gamma_\ell(G)$ executed after $\opa$ on the same process $p$ . And as $\opb\notin \Lambda_\ell(G)$, there exists a vertex $\opd\in \Lambda_\ell(G)$ executed before $\opb$ on the same process $q$. By definition of $\Lambda_\ell(G)$ and $\Gamma_\ell(G)$ we have $\opa.rtime<\opc.stime<\opd.stime<\opb.stime$.
Having $\opa\rightarrow\opb$ implies that $\opb$ is a direct successor of $\opa$, meaning that $\opb$ is the first operation to occur on its process $q$ after $\opa$ finishes. But there exists $\opd$ on process $q$ which starts after $\opa$ finishes, and occurs before $\opb$. This implies that $\opb$ is not a direct successor of $\opa$. By construction of the algorithm, if $\opb$ is not a direct successor of $\opa$, then the edge $\opa\rightarrow\opb$ cannot exist. This leads to a contradiction.
Thus, there is no edge in $L_\ell(G)\times R_\ell(G)$ that crosses the cut~$\ell$.
\end{proof}

\end{document}